\documentclass[letterpaper, 10 pt, conference]{ieeeconf}
\IEEEoverridecommandlockouts
\let\labelindent\relax
\usepackage{enumitem}
\usepackage{balance}
\usepackage[font={small}]{caption}
\usepackage{subcaption}
\usepackage{array}
\usepackage{textcomp}
\usepackage{mathtools, nccmath}
\usepackage{graphicx}
\usepackage{amsfonts}
\usepackage{amsmath}
\usepackage{amssymb}
\usepackage{algorithm}
\usepackage{algorithmic}
\usepackage{hyperref}
\usepackage{tikz}
\usepackage{arydshln}
\usepackage{multirow}
\usepackage{bm}
\usepackage{epstopdf}
\usepackage{cite}
\usepackage{siunitx}
\usepackage{enumitem}
\usepackage{romannum}
\usepackage{cancel}
\usepackage{microtype}

\usepackage{amsthm}

\newtheorem{theorem}{Theorem}

\newtheorem{remark}{Remark}
\theoremstyle{definition}
\newtheorem{definition}{Definition}
\newtheorem{problem}{Problem}

\title{\LARGE \bf
    Sampling-Aware Control Barrier Functions for Safety-Critical and Finite-Time Constrained Control}

\author{Shuo Liu$^{1}$, Wei Xiao$^{2}$ and Calin A. Belta$^{3}$
\thanks{This work was supported in part by the NSF under grant IIS-2024606 at Boston University.}
\thanks{$^{1}$S. Liu is with the department of Mechanical Engineering, Boston
University, Brookline, MA, USA. 
        {\tt\small liushuo@bu.edu}}%
\thanks{$^{2}$W. Xiao is with Department of  Robotics Engineering, Worcester Polytechnic Institute, MA, USA.
        {\tt\small wxiao3@wpi.edu}}%
\thanks{$^{3}$C. Belta is with the Department of Electrical and Computer Engineering and the Department of Computer Science, University of Maryland, College Park, MD, USA 
        {\tt\small calin@umd.edu}}%
}

\begin{document} 
\maketitle

\begin{abstract}
In safety-critical control systems, ensuring both safety and feasibility under sampled-data implementations is crucial for practical deployment. Existing Control Barrier Function (CBF) frameworks, such as High-Order CBFs (HOCBFs), effectively guarantee safety in continuous time but may become unsafe when executed under zero-order-hold (ZOH) controllers due to inter-sampling effects. Moreover, they do not explicitly handle finite-time reach-and-remain requirements or multiple simultaneous constraints, which often lead to conflicts between safety and reach-and-remain objectives, resulting in feasibility issues during control synthesis. This paper introduces Sampling-Aware Control Barrier Functions (SACBFs), a unified framework that accounts for sampling effects and high relative-degree constraints by estimating and incorporating Taylor-based upper bounds on barrier evolution between sampling instants. The proposed method guarantees continuous-time forward invariance of safety and finite-time reach-and-remain sets under ZOH control. To further improve feasibility, a relaxed variant (r-SACBF) introduces slack variables for handling multiple constraints realized through time-varying CBFs. Simulation studies on a unicycle robot demonstrate that SACBFs achieve safe and feasible performance in scenarios where traditional HOCBF methods fail.
\end{abstract}
\section{Introduction}
Safety is a fundamental consideration in the design and operation of autonomous systems. To ensure safety, extensive research has focused on incorporating safety constraints into optimal control formulations through the use of Barrier Functions (BFs) and Control Barrier Functions (CBFs). Originally introduced in the context of optimization \cite{boyd2004convex}, BFs are Lyapunov-like functions \cite{tee2009barrier} that have been widely employed to establish set invariance \cite{aubin2011viability, prajna2007framework} and to develop multi-objective and multi-agent control strategies \cite{wang2017safety, glotfelter2017nonsmooth}.

CBFs generalize BFs to guarantee the forward invariance of safe sets for affine control systems. When a CBF satisfies Lyapunov-like conditions, system safety is ensured \cite{ames2016control}. By integrating CBFs with Control Lyapunov Functions (CLFs), the CBF-CLF-QP framework formulates safe and optimal control problems as a sequence of Quadratic Programs (QPs) \cite{ames2012control, ames2016control}. Although early formulations were limited to constraints with relative degree one, subsequent extensions such as Exponential CBFs \cite{nguyen2016exponential} and High-Order CBFs (HOCBFs) \cite{xiao2021high} have enabled handling of higher-order constraints. The CBF-CLF-QP framework has been applied across various domains, including rehabilitation robotics \cite{isaly2020zeroing}, adaptive cruise control \cite{liu2023auxiliary, liu2025auxiliary}, humanoid locomotion \cite{khazoom2022humanoid}, and obstacle avoidance in complex environments \cite{liu2025learning, liu2024safety}.

Recent works also employed CBFs or CLFs to enforce time-constrained or task-based specifications on the system trajectories, like \cite{lindemann2018control,xiao2021high2,garg2019control,liu2024auxiliary}. In practice, physical systems evolve in continuous time, while controllers are typically implemented in discrete time, such as zero-order-hold (ZOH) controllers with a fixed sampling period. However, counterexamples can be readily constructed to show that control laws derived from the CBF conditions in \cite{hsu2015control, wang2017safety} may lose their safety guarantees when implemented in discrete time. Conversely, controllers designed using discrete-time CBFs may not guarantee continuous-time safety between sampling instants \cite{yang2020continuous}.

Recently, \cite{cortez2019control} proposed a method for ensuring satisfaction of the continuous-time CBF condition under a ZOH control law by bounding the time derivative of the CBF between sampling instants.
 \cite{yang2019self} developed a self-triggered control framework for safety-critical systems that adaptively determines the next sampling time based on the CBF evolution, thereby reducing conservatism compared to fixed-period sampling.
Subsequently, \cite{breeden2021control} presented a sampled-data formulation that guarantees forward invariance under ZOH implementation by explicitly accounting for inter-sampling effects.
\cite{taylor2022safety} proposed an approach that employs approximate discrete-time models to verify safety between sampling intervals.
More recently, \cite{niu2024sampling} incorporated both sampling and quantization effects into the CBF design to enhance the safety of cyber-physical systems implemented in digital environments.
Despite these advances, existing approaches still have notable limitations.
They primarily address constraints with relative degree one and thus cannot be directly extended to high-order constraints; they do not consider temporal constraints, which are usually realized through time-varying CBFs; and when multiple safety constraints are simultaneously imposed, the resulting optimization problems often suffer from infeasibility.

In this paper, we consider finite-time reach-and-remain requirements, which aim to guarantee that the system states reach a specified target set within a user-defined time horizon and subsequently remain within this set for a prescribed duration. Such requirements naturally arise when decomposing Signal Temporal Logic (STL) specifications into basic control objectives—primarily safety, reachability, and persistence constraints \cite{lindemann2018control,xiao2021high2}. To ensure safety and finite-time reach-and-remain for nonlinear systems with high relative-degree constraints under ZOH control, we propose Sampling-Aware Control Barrier Functions (SACBFs). Specifically, we develop a unified framework of SACBFs that explicitly account for sampling effects and high relative-degree constraints. By incorporating Taylor-based estimates of the upper bounds on the barrier function evolution between sampling instants, the proposed SACBFs ensure continuous-time forward invariance of the sets for safety and finite-time reach-and-remain under ZOH control. To enhance feasibility when multiple constraints are simultaneously enforced, a relaxation variable is introduced to each constraint without compromising forward invariance. Furthermore, we define two candidate SACBF formulations to handle aforementioned requirements: one designed for safety with time-invariant structure, and the other for finite-time reach-and-remain realized through time-varying CBFs. The effectiveness and flexibility of the proposed framework are demonstrated through case studies on a unicycle robot, showing that the relaxed SACBF achieves safe and feasible performance where traditional HOCBF methods fail.
\section{Preliminaries}
\label{sec:Preliminaries}
 Consider an affine control system expressed as 
 \begin{equation}
 \label{eq: affine-system}
\dot{\boldsymbol{x}}=f(\boldsymbol{x})+g(\boldsymbol{x})\boldsymbol{u}, 
 \end{equation}
 where $\boldsymbol{x}\in \mathbb{R}^{n},  f:\mathbb{R}^{n}\to\mathbb{R}^{n}$ and $g:\mathbb{R}^{n}\to\mathbb{R}^{n\times q}$ are locally Lipschitz. The control is $\boldsymbol{u}\in \mathcal U\subset \mathbb{R}^{q}$, where the control constraint set $\mathcal U$ is defined as 
  \begin{equation}
  \label{eq: input bounds}
  \mathcal U \coloneqq \{\boldsymbol{u}\in \mathbb{R}^{q}:\boldsymbol{u}_{min}\le \boldsymbol{u}\le \boldsymbol{u}_{max} \}, 
  \end{equation}
 with $\boldsymbol{u}_{min}, \boldsymbol{u}_{max}\in \mathbb{R}^{q}$ (vector inequalities are interpreted componentwise).\\
\begin{definition}[Lipschitz continuity~\cite{rockafellar1998variational}]
\label{def:Lipschitz continuity}
A function $h: \mathbb{R}^n \to \mathbb{R}^m$ is said to be Lipschitz continuous on a set $\mathcal{D} \subseteq \mathbb{R}^n$ if there exists a constant $L \geq 0$ such that
\begin{equation}
\| h(\mathbf{x}) - h(\mathbf{y}) \| \leq L \| \mathbf{x} - \mathbf{y} \|, \quad \forall \mathbf{x},\mathbf{y} \in \mathcal{D}.
\end{equation}
The smallest such $L$ is called the Lipschitz constant of $h$.
\end{definition}

\begin{definition}[Class $\kappa$ function~\cite{Khalil:1173048}]
 A continuous function $\alpha:[0,a)\to[0,+\infty],a>0$ is called a class $\kappa$ function if it is strictly increasing and $\alpha(0)=0$.
\end{definition}

\begin{definition} A set $\mathcal C\subset \mathbb{R}^{n}$ is forward invariant for system \eqref{eq: affine-system} if its solutions for some $\boldsymbol{u} \in \mathcal U$ starting from any $\boldsymbol{x}(0) \in \mathcal C$ satisfy $\boldsymbol{x}(t) \in \mathcal C, \forall t \ge 0.$
\end{definition}

\begin{definition} The relative degree of a differentiable function $b:\mathbb{R}^{n}\to\mathbb{R}$ is the number of times we need to differentiate it along system dynamics \eqref{eq: affine-system} until any component of $\boldsymbol{u}$ explicitly shows in the corresponding derivative. 
\end{definition}

In this paper, a \textbf{safety requirement} is defined as $b(\boldsymbol{x},t)\ge0$, and \textbf{safety} is the forward invariance of the set $\mathcal C\coloneqq \{\boldsymbol{x}\in\mathbb{R}^{n}:b(\boldsymbol{x},t)\ge 0\}$. The relative degree of time-varying function $b$ is also referred to as the relative degree of safety requirement $b(\boldsymbol{x},t) \ge 0$.  For $b(\boldsymbol{x},t)\ge 0$ with relative degree $m, \ b:\mathbb{R}^{n}\to\mathbb{R}$ and $\psi_{0}(\boldsymbol{x},t)\coloneqq b(\boldsymbol{x},t),$ we define a sequence of functions $\psi_{i}:\mathbb{R}^{n}\to\mathbb{R},\ i\in \{1,...,m\}$ as
 \begin{equation}
 \label{eq: HOCBFs}
 \psi_{i}(\boldsymbol{x},t)\coloneqq\dot{\psi}_{i-1}(\boldsymbol{x},t)+\alpha_{i}(\psi_{i-1}(\boldsymbol{x},t)),\ i\in \{1,...,m\}, 
 \end{equation}
where $\alpha_{i}(\cdot ),\ i\in \{1,...,m\}$ denotes a $(m-i)^{th}$ order differentiable class $\kappa$ function.  A sequence of sets $\mathcal C_{i}$ are then defined based on \eqref{eq: HOCBFs} as 
\begin{equation}
\label{eq: HOCBF sets}
\mathcal C_{i}\coloneqq \{\boldsymbol{x}\in\mathbb{R}^{n}:\psi_{i-1}(\boldsymbol{x},t)\ge 0\}, \ i\in \{1,...,m\}. 
\end{equation}
\begin{definition}[Time-varying High Order Control Barrier Function (HOCBF)~\cite{xiao2021high}]\label{def: HOCBF} Let $\psi_{i}(\boldsymbol{x},t),\ i\in \{1,...,m\}$ be defined by \eqref{eq: HOCBFs} and $\mathcal C_{i},\ i\in \{1,...,m\}$ be defined by \eqref{eq: HOCBF sets}. A function $b:\mathbb{R}^{n}\to\mathbb{R}$ is a time-varying High Order Control Barrier Function (HOCBF) with relative degree $m$ for system \eqref{eq: affine-system} if there exist $(m-i)^{th}$ order differentiable class $\kappa$ functions $\alpha_{i},\ i\in \{1,...,m\}$ such that
\begin{equation}
\label{eq: highest HOCBF}
 \begin{split}
 \sup_{\boldsymbol{u}\in \mathcal U}[L_{f}^{m}b(\boldsymbol{x},t)+L_{g}L_{f}^{m-1}b(\boldsymbol{x},t)\boldsymbol{u}+\frac{\partial ^{m}b(\boldsymbol{x},t)}{\partial t^{m}}+\\O(b(\boldsymbol{x},t))
 +  \alpha_{m}(\psi_{m-1}(\boldsymbol{x},t))]\ge 0,
 \end{split}
 \end{equation}
 $\forall \boldsymbol{x}\in \mathcal C_{1}\cap,...,\cap \mathcal C_{m}.$ 
\end{definition}
In \eqref{eq: highest HOCBF}, $L_{f}^{m}$ denote $m$-th order Lie derivatives along $f$; $L_{g}$ is the Lie derivative along $g$. We use $O(\cdot)=\sum_{i=1}^{m-1}L_{f}^{i}(\alpha_{m-1}\circ\psi_{m-i-1})(\boldsymbol{x},t) $ and assume that $L_{g}L_{f}^{m-1}b(\boldsymbol{x},t)\boldsymbol{u}\ne0$ on the boundary of set $\mathcal C_{1}\cap,...,\cap \mathcal C_{m}.$ 

\begin{theorem}[Safety Guarantee~\cite{xiao2021high}]
\label{thm:safety-guarantee}
Given a HOCBF $b(\boldsymbol{x},t)$ from Def. \ref{def: HOCBF} with corresponding sets $\mathcal{C}_{0}, \dots,\mathcal {C}_{m-1}$ defined by \eqref{eq: HOCBF sets}, if $\boldsymbol{x}(0) \in \mathcal {C}_{0}\cap \dots \cap \mathcal {C}_{m-1},$ then any Lipschitz controller $\boldsymbol{u}$ that satisfies the inequality in \eqref{eq: highest HOCBF}, $\forall t\ge 0$ renders $\mathcal {C}_{0}\cap \dots \cap \mathcal {C}_{m-1}$ forward invariant for system \eqref{eq: affine-system}, $i.e., \boldsymbol{x} \in \mathcal {C}_{0}\cap \dots \cap \mathcal {C}_{m-1}, \forall t\ge 0.$
\end{theorem}
The existing works~\cite{nguyen2016exponential, xiao2021high} leverage HOCBFs~\eqref{eq: highest HOCBF} to handle systems with high relative degree, combining them with quadratic cost functions to formulate safety-critical optimization-based controllers. In these frameworks, the continuous-time problem is discretized over sampling intervals $[t_k, t_{k+1})$, where the state $\boldsymbol{x}(t_k)$ is measured and a quadratic program (QP) with HOCBF constraints is solved to obtain the optimal control $\boldsymbol{u}^\ast(t_k)$, which is applied in a zero-order-hold (ZOH) manner while the system evolves according to~\eqref{eq: affine-system}. Because controller updates occur only at discrete sampling times, the system evolves open-loop between updates, potentially causing inter-sampling effects~\cite{singletary2020control}, where transient constraint violations arise even if the HOCBF conditions hold at the sampled instants. The severity of this effect depends on several factors, including the sampling period, the Lipschitz continuity of the system dynamics, and the steepness of the barrier functions. Specifically, when the sampling interval is not sufficiently small or the system exhibits fast dynamics, the state trajectory may temporarily exit the safe set between samples, leading to safety degradation despite satisfying the HOCBF constraints at discrete times.

\section{Problem Formulation and Approach}
\label{sec:Problem Formulation and Approach}

The objective is to develop a control strategy for the system in \eqref{eq: affine-system} that ensures continuous-time reach-and-remain of a desired target within prescribed time horizons, minimizes energy consumption, guarantees continuous-time safety, and satisfies the input constraints \eqref{eq: input bounds}.

\textbf{Finite-Time Reach-and-Remain Requirement:} 
Given time bounds $0 \le T_1 < T_2 \le \infty$, 
the state of system~\eqref{eq: affine-system} is required to reach a closed set 
$S \coloneqq \{\boldsymbol{x}\in\mathbb{R}^{n} \mid h(\boldsymbol{x}) \ge 0\}$,
defined by a continuously differentiable function $h:\mathbb{R}^{n}\to\mathbb{R}$,
within the horizon $T_1$.
That is, for a given initial condition $\boldsymbol{x}(0)\in\mathbb{R}^{n}$,
there must exist a time $t_r \in [0,\,T_1]$ such that $\boldsymbol{x}(t_r)\in S$.
Furthermore, after reaching the set, the state is required to remain in $S$
for all subsequent times over the interval $[t_r,\,T_2]$, i.e.,
$\boldsymbol{x}(t)\in S$ for all $t \in [t_r,\,T_2]$.

\textbf{Safety Requirement:} System \eqref{eq: affine-system} should always satisfy a safety requirement of the form: 
\begin{equation}
\label{eq:Safety constraint}
b(\boldsymbol{x}(t))\ge 0, \boldsymbol{x} \in \mathbb{R}^{n}, \forall t \in [0, T_2],
\end{equation}
where $b:\mathbb{R}^{n}\to\mathbb{R}$ is assumed to be a continuously differentiable function.

\textbf{Control Limitation Requirement:} The controller $\boldsymbol{u}$ should always satisfy \eqref{eq: input bounds} for all $t \in [0, T_2].$

A control policy is said to be \textbf{feasible} if all constraints arising from the aforementioned requirements are satisfied and remain mutually non-conflicting over the time interval $[0, T_2]$. 
In this paper, we consider the following problem:

\begin{problem}
\label{prob:SACC-prob}
Determine a feasible control policy for system~\eqref{eq: affine-system} that
minimizes the following cost:
\begin{equation}
\label{eq:cost-function-1}
\begin{split}
 J(\boldsymbol{u}(t))=\int_{0}^{T_2} 
 D(\left \| \boldsymbol{u}(t) \right \|)dt,
\end{split}
\end{equation}
while satisfying  all previously stated requirements.
\end{problem}

In Eq. (\ref{eq:cost-function-1}), $\left \| \cdot \right \|$ denotes the 2-norm of a vector and $D(\cdot)$ is a strictly increasing function of its argument. This problem setup is directly motivated by STL, as spatio-temporal STL specifications naturally reduce to finite-time reach-and-remain tasks. Therefore, solving Problem~\ref{prob:SACC-prob} directly enables satisfaction of this class of STL requirements.

As mentioned in Sec. \ref{sec:Preliminaries}, when Problem~\ref{prob:SACC-prob} is reduced to a sequence of QPs and solved at discrete time steps, the corresponding discretization introduces inter-sampling effects, where the control input is held constant in a ZOH manner between updates and the system evolves open-loop, potentially leading to transient requirement violations. Existing sampled-data CBF formulations that address inter-sampling issues cannot effectively handle time-varying or high-order CBFs, and when multiple CBF constraints are enforced simultaneously, the tightened inequalities may conflict, leading to infeasibility of the optimization problem. These limitations motivate the development of new formulations that can maintain continuous-time safety, reach-and-remain and feasibility for time-varying high-order CBFs under sampled implementations.

 \textbf{Approach:} In this paper, we introduce Sampling-Aware Control Barrier Functions (SACBFs) that explicitly address both the inter-sampling and high relative degree issues. By incorporating Taylor-based upper bounds on the barrier function’s evolution between sampling instants, SACBFs ensure continuous-time forward invariance of corresponding sets under ZOH control. Moreover, a relaxation variable is added to each SACBF constraint to improve feasibility when multiple SACBFs are jointly used to enforce the desired requirements.

 \section{Sampling-Aware Control Barrier Functions}
\label{sec:Sampling-Aware Control Barrier Functions}
In this section, we develop the framework of Sampling-Aware Control Barrier Functions (SACBFs), which systematically accounts for inter-sampling effects and high relative degree constraints to ensure continuous-time forward invariance of desired sets under ZOH control.
\subsection{Sampling-Aware Control Barrier Functions}
\label{subsec:Sampling-Aware Control Barrier Functions}
Let the sampling period be $\Delta t > 0$, and consider the affine control system \eqref{eq: affine-system}
evolving on a compact set $\mathcal{X}$, where the control input $\boldsymbol{u}\in\mathcal{U}$ is held constant
within each sampling interval $[t_k,\,t_{k+1})$ under a ZOH implementation, i.e., 
\begin{equation}
\boldsymbol{u}(t) = \boldsymbol{u}(t_k), \quad t \in [t_k, t_{k+1}).
\end{equation}
Consider functions $\psi_i(\boldsymbol{x},t)$ defined in \eqref{eq: HOCBFs}, where each $\alpha_i$ is defined as $\alpha_i(\psi_{i-1}(\boldsymbol{x},t)) = \lambda_i \psi_{i-1}(\boldsymbol{x},t)^{\eta_{i}}$ for some $\lambda_i, \eta_{i}> 0$. Integrating $\dot{\psi}_{m-1}(\boldsymbol{x},t) \ge -\lambda_m \psi_{m-1}(\boldsymbol{x},t)^{\eta_{m}}$ and using the comparison lemma in \cite{Khalil:1173048}, we have 
\begin{small}
\begin{equation}
\label{eq: CBF integral}
\begin{split}
\psi_{m-1}(\boldsymbol{x}(t),t)
\ge \psi_{m-1}(\boldsymbol{x}(t_0),t_0)
\, e^{-\lambda_m (t - t_0)}, \eta_{i}=1,\\
\psi_{m-1}(\boldsymbol{x}(t),t) \ge
\Bigl[\psi_{m-1}(\boldsymbol{x}(t_0),t_0)^{1-\eta_{m}}-\\
\lambda_{m}(1-\eta_{m})\,(t-t_0)\Bigr]_+^{\;\frac{1}{1-\eta_{m}}},  \eta_{i}\ne 1,
\end{split}
\end{equation}
\end{small}
where $[a]_+ := \max\{a,\,0\}.$ From the above expression, it can be seen that the right-hand side of the inequality 
is always nonnegative if $\psi_{m-1}(\boldsymbol{x}(t_0),t_0)\ge 0$. Therefore, Eq.~\eqref{eq: CBF integral} can be compactly represented as follows:
\begin{equation}
\label{eq: CBF integral2}
\psi_{m-1}(\boldsymbol{x}(t),t) \ge \mathcal{L}(\psi_{m-1}(\boldsymbol{x}(0),0),\lambda_{m},\eta_{m},t-t_0)\ge 0,
\end{equation}
where $\psi_{m-1}(\boldsymbol{x}(0),0)\ge0,\lambda_{m}>0,\eta_{m}>0$. For notational simplicity, we denote $\psi_{m-1}(t) \Longleftrightarrow  \psi_{m-1}(\boldsymbol{x}(t),t)$, $\mathcal{L}(t_0, t-t_0)\Longleftrightarrow \mathcal{L}(\psi_{m-1}(\boldsymbol{x}(0),0),\lambda_{m},\eta_{m},t-t_0)$ and the same convention is used hereafter. For the interval $[t_k, t_{k+1})$ starting from $t_k$, based on \eqref{eq: CBF integral2}, we have 
\begin{equation}
\label{eq: positive guarantee}
\psi_{m-1}(t_k+\bigtriangleup t) \ge \mathcal{L}(t_{k}, \bigtriangleup t)\ge0.
\end{equation}
Expanding $\psi_{m-1}(t_k + \Delta t)$ around $t_k$ up to the second order using a Taylor-based expansion yields 
\begin{equation}
\label{eq: second-order term}
\begin{split}
\psi_{m-1}(t_k+\bigtriangleup t)
= \psi_{m-1}(t_k)
+ \bigtriangleup t\dot{\psi}_{m-1}(t_k) 
+ \\
\frac{{(\bigtriangleup t)}^{2}}{2}\ddot{\psi}_{m-1}(\xi),
\end{split}
\end{equation}
where $\xi \in [t_k,\, t_k + \bigtriangleup t]$ is some intermediate time, and the last term in the above corresponds to the Lagrange remainder. To calculate $\ddot{\psi}_{m-1}(\xi)$, we apply the chain rule of differentiation along dynamic system \eqref{eq: affine-system} and obtain
\begin{small}
\begin{equation}
\label{eq: Taylor expansion}
\begin{split}
\ddot{\psi}_{m-1}(t)
=\dot{\boldsymbol{x}}^{\top}(\nabla_{\boldsymbol{x}}^2 \psi_{m-1}(t))\dot{\boldsymbol{x}}+ \nabla_{\boldsymbol{x}}& \psi_{m-1}(t)^{\top}\!
\Big(f_{\boldsymbol{x}}\dot{\boldsymbol{x}}+ \\
g_{\boldsymbol{x}}\dot{\boldsymbol{x}}\boldsymbol{u}
  + f_{t}
  + g_{t}\boldsymbol{u}
+\cancel{g(\boldsymbol{x},t)\dot{\boldsymbol{u}}}
\Big)+& 2\nabla_{\boldsymbol{x}t}^2 \psi_{m-1}(t)^{\top}\dot{\boldsymbol{x}} \\
&  + \nabla_{tt}^2 \psi_{m-1}(t),
\end{split}
\end{equation}
\end{small}
where $f_{\boldsymbol{x}}\coloneqq \tfrac{\partial f}{\partial \boldsymbol{x}}(\boldsymbol{x},t)$, $f_{t}\coloneqq \tfrac{\partial f}{\partial t}(\boldsymbol{x},t)$, $g_{\boldsymbol{x}}\coloneqq \tfrac{\partial g}{\partial \boldsymbol{x}}(\boldsymbol{x},t)$, $g_{t}\coloneqq \tfrac{\partial g}{\partial t}(\boldsymbol{x},t)$, $\nabla_{\boldsymbol{x}}\psi_{m-1}(t)
:= \frac{\partial}{\partial \boldsymbol{x}}\left(\psi_{m-1}(\boldsymbol{x},t)\right)$, $\nabla_{\boldsymbol{x}}^{2}\psi_{m-1}(t)
:= \frac{\partial}{\partial \boldsymbol{x}}\left(\nabla_{\boldsymbol{x}} \psi_{m-1}(t)\right)$, $\nabla_{\boldsymbol{x}t}^2 \psi_{m-1}(t)
:= \frac{\partial}{\partial t}\!\left(\nabla_{\boldsymbol{x}} \psi_{m-1}(t)\right)$, $\nabla_{tt}^2 \psi_{m-1}(t)
:= \frac{\partial^2 \psi_{m-1}(t)}{\partial t^2}$, $\dot{\boldsymbol{u}}=\textbf{0}$. In \eqref{eq: Taylor expansion}, the arguments $\boldsymbol{x}(t)$ and $\boldsymbol{u}(t)$ in $\ddot{\psi}_{m-1}(t)$ are omitted for notational simplicity. Because $f$, $g$, and $\psi_{m-1}$ are continuously differentiable and the system 
trajectory $\boldsymbol{x}(t)$ evolves within a compact set $\mathcal{X}$,
all terms in \eqref{eq: Taylor expansion} are continuous with respect to $t$.
Hence, $\ddot{\psi}_{m-1}(t)$ is continuous on the closed interval 
$[t_k,\,t_k+\Delta t]$. 
By the extreme value theorem in \cite{rudin1976principles}, there exists a finite constant $\bar{M}_k>0$ such that 
$|\ddot{\psi}_{m-1}(t)| \le \bar{M}_k,  \forall\,t\in[t_k,\,t_k+\Delta t]$.

\begin{definition}[Sampling-Aware Control Barrier Function (SACBF)]
\label{def:SACBF}
Let $\psi_{i}(\boldsymbol{x},t),\ i\in \{1,\ldots,m\}$ be defined by~\eqref{eq: HOCBFs},  $\mathcal{C}_{i},\ i\in \{1,\ldots,m\}$ be defined by~\eqref{eq: HOCBF sets} and $\mathcal{L}(t_k, \Delta t)$ be defined by \eqref{eq: CBF integral}, \eqref{eq: CBF integral2}. Consider the affine control system~\eqref{eq: affine-system}, which evolves within a compact set $\mathcal{X}\subset\mathbb{R}^n$, where the control input $\boldsymbol{u}(t)$ is held constant, i.e., $\boldsymbol{u}(t)=\boldsymbol{u}(t_k)$ for $t\in[t_k,t_k +\Delta t)$, $k\in\{0, 1, \ldots\}$ under a zero-order-hold (ZOH) implementation. A function $b:\mathbb{R}^{n}\to\mathbb{R}$ is called a Sampling-Aware Control Barrier Function (SACBF) with relative degree $m$ for system~\eqref{eq: affine-system} if there exist $(m\!-\!i)^{\text{th}}$-order differentiable class-$\kappa$ functions $\alpha_i(\psi_{i-1}(\boldsymbol{x},t)) = \lambda_i \psi_{i-1}(\boldsymbol{x},t)^{\eta_{i}}$, with $\lambda_i>0$ and $\eta_{i}>0$, $i\in \{1,\ldots,m\}$, such that
\begin{equation}
\label{eq: SACBF}
\begin{split}
\sup_{\boldsymbol{u}\in \mathcal U}[L_f \psi_{m-1}(t_k) + L_g \psi_{m-1}(t_k) \boldsymbol{u} + \frac{\partial \psi_{m-1}(t_k)}{\partial t}]
\ge \\
\frac{\mathcal{L}(t_k, \Delta t)-\psi_{m-1}(t_k)}{\Delta t}
+\tfrac{1}{2}\bar{M}_k \Delta t,
\end{split}
\end{equation}
$\forall \boldsymbol{x}\in \mathcal C_{1}\cap,...,\cap \mathcal C_{m}$. $\bar{M}_k>0$ denotes the upper bound of the second derivative term (i.e., $|\ddot{\psi}_{m-1}(t)|$ in \eqref{eq: second-order term}) over the sampling interval $[t_k, t_k+\Delta t]$.
\end{definition}
As $\Delta t \to 0$, the SACBF condition in \eqref{eq: SACBF} reduces to the HOCBF condition from \eqref{eq: HOCBFs}, since $\frac{\mathcal{L}(t_k,\Delta t)-\psi_{m-1}(t_k)}{\Delta t}\to -\alpha_m(\psi_{m-1}(t_k))$ and the correction term vanishes.
\begin{theorem}
\label{thm:SACBF-safety-guarantee}
Given a SACBF $b(\boldsymbol{x})$ as defined in Def.~\ref{def:SACBF}, with the associated sets $\mathcal{C}_{0},\dots,\mathcal{C}_{m-1}$ defined by~\eqref{eq: HOCBF sets}, 
suppose the initial condition satisfies $\boldsymbol{x}(t_0)\in\mathcal{C}_{0}\cap\dots\cap\mathcal{C}_{m-1}$. 
Then, under ZOH control with sampling period $\Delta t>0$, any locally Lipschitz controller $\boldsymbol{u}(t)$ that satisfies the SACBF condition~\eqref{eq: SACBF} for all $t\ge0$ renders the set 
$\mathcal{C}_{0}\cap\dots\cap\mathcal{C}_{m-1}$ forward invariant for the system~\eqref{eq: affine-system}, $i.e., \boldsymbol{x} \in \mathcal {C}_{0}\cap \dots \cap \mathcal {C}_{m-1}, \forall t\ge 0.$
\end{theorem}
\begin{proof}
Starting from $[t_0,t_0+\Delta t]$, Taylor expansion in \eqref{eq: second-order term} together with
$\ddot{\psi}_{m-1}(\xi)\ge -\bar M_0$ yields
\begin{equation}
\label{eq: proof1}
\begin{split}
\psi_{m-1}(t_0+\Delta t)\ge \psi_{m-1}(t_0)+\Delta t\,\dot{\psi}_{m-1}(t_0)
-\frac{\Delta t^2}{2}\bar M_0.
\end{split}
\end{equation}
By \eqref{eq: SACBF} at $t_k=t_0$,
\begin{equation}
\label{eq: proof2}
\begin{split}
\dot{\psi}_{m-1}(t_0)\ge
\frac{\mathcal{L}(t_0,\Delta t)-\psi_{m-1}(t_0)}{\Delta t}
+\frac{1}{2}\bar M_0\Delta t.
\end{split}
\end{equation}
Substituting \eqref{eq: proof2} into the previous inequality \eqref{eq: proof1} gives
$\psi_{m-1}(t_0+\Delta t)\ge \mathcal{L}(t_0,\Delta t)\ge 0$, which indicates that satisfying condition \eqref{eq: SACBF} guarantees the nonnegativity of 
$\psi_{m-1}(t_0+\Delta t)$ at the end of the sampling interval. Let $\tau \in [0,\,\Delta t]$, we want
\begin{equation}
\label{eq: proof3}
\begin{split}
\psi_{m-1}(t_0+\tau)\ge \psi_{m-1}(t_0)
+ \tau \dot{\psi}_{m-1}(t_0)
- 
\frac{{\tau}^{2}}{2}\bar{M}_0\ge 0.
\end{split}
\end{equation}

Substituting \eqref{eq: proof2} into \eqref{eq: proof3}, we want 
\begin{equation}
\label{eq: proof4}
\begin{split}
\psi_{m-1}&(t_0+\tau) \ge l(\tau)=
\psi_{m-1}(t_0)
+ \\ &\tau \Big(\frac{\mathcal{L}(t_{0}, \bigtriangleup t)-\psi_{m-1}(t_0)}{\Delta t}
+ \tfrac{1}{2}\bar{M}_0 \Delta t\Big)
- \tfrac{1}{2}\bar{M}_0 {\tau}^{2}
 \ge 0.
\end{split}
\end{equation}
Define $A=\frac{\mathcal{L}(t_{0}, \bigtriangleup t)-\psi_{m-1}(t_0)}{\Delta t}, B=\tfrac{1}{2}\bar{M}_0 \Delta t$. The function $l(\tau)$ is a quadratic function of $\tau$ that opens downward. We assume that $\Delta t$ is larger than its largest root, as given below:
\begin{equation}
\label{eq: proof5}
\begin{split}
\Delta t > \frac{-(A+B)-\sqrt{(A+B)^{2}+2\bar{M}_{0}\psi_{m-1}(t_0)}}{-\bar{M}_0}.
\end{split}
\end{equation}
By rewriting \eqref{eq: proof5}, we have
\begin{small}
\begin{equation}
\label{eq: proof6}
\begin{split}
-\bar{M}_0\Delta t +A +B = A-B < -\sqrt{(A+B)^{2}+2\bar{M}_{0}\psi_{m-1}(t_0)},
\end{split}
\end{equation}
\end{small}
where $\psi_{m-1}(t_0)\ge 0$. Following the above equation, we have that $A-B \leq 0$, 
then we square both sides of inequality~\eqref{eq: proof6} to obtain $(A - B)^2 \ge (A + B)^{2} + 2\bar{M}_{0}\psi_{m-1}(t_0)$. Rearranging terms gives $-2AB \ge \bar{M}_{0}\psi_{m-1}(t_0)$. Substituting $A$ and $B$ into this inequality, we have $\bar{M}_{0}\mathcal{L}(t_{0}, \Delta t) \le 0$. However, since $\bar{M}_{0}\mathcal{L}(t_{0}, \Delta t) \ge 0$ clearly holds, the assumption \eqref{eq: proof5} is invalid, and therefore $\Delta t$ must be less than or equal to the larger root of the quadratic function $l(\tau)$. Consequently, since the quadratic function $l(\tau)$ remains nonnegative for all $\tau \in [0, \Delta t]$ as long as $\Delta t$ does not exceed its larger root, ensuring that $\Delta t$ is less than or equal to the larger root of $l(\tau)$ guarantees that $l(\tau)\ge 0$ holds throughout the interval. Therefore, satisfying inequality~\eqref{eq: SACBF} ensures that $\psi_{m-1}(t_0+\tau)$ remains nonnegative for all $\tau \in [0, \Delta t]$.

We now have proved that $\psi_{m-1}(t) \ge 0, \forall t \in[t_0, t_0+\Delta t]$. By shifting the time interval to $[t_1, t_1+\Delta t]$ and repeating the same proof, we obtain $\psi_{m-1}(t) \ge 0, \forall t \in[t_1, t_1+\Delta t]$. By iterating this process, it follows that $\psi_{m-1}(t) \ge 0$ for all $t\ge 0$. With this result and by following the safety guarantee proof of HOCBFs in~\cite{xiao2021high}, we can further derive $\psi_{m-2}(t) \ge 0, \dots, \psi_{0}(t)=b(\boldsymbol{x},t) \ge 0$. Therefore, the intersection of the sets $\mathcal{C}_{0}, \dots, \mathcal{C}_{m-1}$ is forward invariant.
\end{proof}
\begin{remark}
\label{rem:choice of classk}
To simplify the computation, in Eq. \eqref{eq: CBF integral} and Def. \ref{def:SACBF}, we define the class $\kappa$ functions in the form $\alpha_i(\psi_{i-1}(\boldsymbol{x},t)) = \lambda_i \psi_{i-1}(\boldsymbol{x},t)^{\eta_i}$. However, the definition of a class~$\kappa$ function can be made more general, for example, by defining it as a polynomial function of $\psi_{i-1}(\boldsymbol{x},t)$. As long as all coefficients and exponents of the polynomial are positive, the Def. \ref{def:SACBF} and Thm. \ref{thm:SACBF-safety-guarantee} of SACBFs remain valid.
\end{remark}
\subsection{Relaxation of SACBF Constraint for Feasibility}
\label{subsec:relaxed SACBF}
Since Eq. \eqref{eq: SACBF} serves as a constraint for solving the optimal input, it may conflict with the input bounds \eqref{eq: input bounds} or other SACBF constraints, resulting in infeasibility. Hence, a relaxation variable is added to improve feasibility without compromising the forward invariance guarantee of the desired set. In Eq. \eqref{eq: positive guarantee}, we add a slack variable $\omega$ and obtain
\begin{equation}
\label{eq: positive guarantee2}
\psi_{m-1}(t_k+\bigtriangleup t) \ge \omega_k\mathcal{L}(t_{k}, \bigtriangleup t)\ge0,
\end{equation}
where $\omega_k \ge 0.$ Due to the modification in Eq.~\eqref{eq: positive guarantee2}, the corresponding term $\mathcal{L}(t_{k}, \bigtriangleup t)$ in Eq.~\eqref{eq: SACBF} is replaced by $\omega_k\mathcal{L}(t_{k}, \bigtriangleup t)$. The slack variable $\omega$ is typically assigned a reference value of $1$ by adding a weight scalar $q>0$ to the quadratic penalty $(\omega-1)^{2}$ in the cost function~\eqref{eq:cost-function-1}, thereby ensuring that Eqs.~\eqref{eq: positive guarantee2} and~\eqref{eq: SACBF} remain unchanged under nominal conditions. When no input can satisfy the SACBF constraint~\eqref{eq: SACBF}, $\omega$ is allowed to decrease within the range $[0,1]$ to relax the constraint and enhance feasibility. Moreover, since $\omega_k \mathcal{L}(t_{k}, \bigtriangleup t) \ge 0$, introducing the slack variable does not affect the non-negativity of $\psi_{m-1}(t_k+\Delta t)$. Therefore, this relaxation enhances feasibility without compromising the forward invariance guarantee.
\subsection{Estimation of the Taylor-Based Upper Bound}
\label{subsec:Taylor-Based Upper Bound}
For notational simplicity, we define $F := f + g\,\boldsymbol{u}$, which represents the affine vector field under a fixed control input $\boldsymbol{u}$. Then, the instantaneous bound of $|\ddot{\psi}_{m-1}(t)|$ in \eqref{eq: second-order term} can be written as
\begin{small}
\begin{equation}
\label{eq: bound1}
\begin{split}
\Phi :=
\|\nabla_{\boldsymbol{x}}^2\psi_{m-1}\|\|F\|^2
+ \|\nabla_{\boldsymbol{x}}\psi_{m-1}\|
  \big(\|f_{\boldsymbol{x}}\|+\|g_{\boldsymbol{x}}\|\|\boldsymbol{u}\|\big) \|F\|\\
+ 2\|\nabla_{\boldsymbol{x}t}^2\psi_{m-1}\|\|F\|
+ \|\nabla_{tt}^2\psi_{m-1}\|
\\+ \|\nabla_{\boldsymbol{x}}\psi_{m-1}\|\big(\|f_t\|+\|g_t\|\|\boldsymbol{u}\|\big).
\end{split}
\end{equation}
\end{small}
This follows from taking the norm of each term in $\ddot{\psi}_{m-1}(t)$ and applying the triangle inequality. To avoid taking a conservative supremum over the entire state space $\mathcal{X}$,
we restrict the evaluation to a local trajectory tube 
\begin{equation}
\label{eq: tube}
\begin{split}
\mathcal{R}_k(\Delta t):=\big\{\boldsymbol{x}:\ \|\boldsymbol{x}-\boldsymbol{x}(t_k)\|\le
\rho(\Delta t)\big\},\\
\rho(\Delta t):=\frac{e^{L_F \Delta t}-1}{L_F}\,\|F(\boldsymbol{x}(t_k),\boldsymbol{u})\|,
\end{split}
\end{equation}
where $L_F$ is the local Lipschitz constant of $F(\cdot)$ with respect to $\boldsymbol{x}$ 
uniformly over $t\in[t_k,t_k+\Delta t]$.
If $L_F=0$, then $\rho(\Delta t)=\Delta t\|F(\boldsymbol{x}(t_k),\boldsymbol{u})\|$.  This local trajectory tube bound follows from the $L_F$ and the integral form of 
Grönwall’s inequality \cite{{Khalil:1173048}}.
The local Taylor-based upper bound is therefore defined as
\begin{equation}
\label{eq: bound2}
\begin{split}
\bar M_{k}(\Delta t):= 
\sup_{\substack{t\in[t_k,t_k+\Delta t]\\ (\boldsymbol{x},\boldsymbol{u})\in\mathcal{R}_k(\Delta t)\times\mathcal{U}}} 
\Phi(\boldsymbol{x},\boldsymbol{u},t), \\
|\ddot{\psi}_{m-1}(t)|\ \le\ \bar M_{k}(\Delta t),\quad
\forall\,t\in[t_k,t_k+\Delta t].
\end{split}
\end{equation}

\textbf{Computational Approximation:}
When evaluating $\bar M_{k}(\Delta t)$ exactly is computationally expensive, 
we approximate it using $r$ sample points based on the Gauss--Legendre nodes \cite{press2007numerical}:
\begin{equation}
\label{eq: node1}
\begin{split}
\hat M_{k} := \max_{i=1,\dots,r} \Phi(\boldsymbol{x}_i,\boldsymbol{u}_i,t_i),~
t_i=t_k+\gamma_i \Delta t,\\ \boldsymbol{x}_i\approx\boldsymbol{x}(t_i), ~\boldsymbol{u}_i\approx\boldsymbol{u}(t_i),
\end{split}
\end{equation}
and
\begin{equation}
\label{eq: bound3}
\begin{split}
\bar M_{k}(\Delta t)\ \le\ \hat M_{k}\ +\ L_{\Phi}(\Delta_{\boldsymbol{x}}+\Delta_{\boldsymbol{u}}),
\end{split}
\end{equation}
where $\{\gamma_i\}$ are the $r$-point Gauss--Legendre quadrature nodes ($r=3$--$5$),
$\Delta_{\boldsymbol{x}}$ and $\Delta_{\boldsymbol{u}}$ denote the maximum variations of the state and control
between adjacent nodes (approximately $\frac{\rho(\Delta t)}{2}$ and $\frac{|\boldsymbol{u}_{i+1}-\boldsymbol{u}_i|}{2}$, respectively),
and $L_{\Phi}$ is the local Lipschitz constant of $\Phi$
with respect to both $\boldsymbol{x}$ and $\boldsymbol{u}$. Gauss--Legendre nodes provide good coverage of smooth functions on a finite interval,
so $\hat M_{k}$ offers a tight and tractable estimate of the local supremum.
\subsection{SACBFs Design for Safety and Finite-Time Reach-and-Remain}
\label{subsec:SACBF design}
To formalize the safety and finite-time reach-and-remain requirements, we define two candidate SACBF formulations based on selected state components.

For the \textbf{safety requirement}, we specify the states $\boldsymbol{x}_{i,s}\in\mathbb{R}^{n_i}$ that must remain within a fixed safe region centered at $\boldsymbol{x}_{i,s_0}$ with a constant boundary radius $r_{i,s_{0}}>0$. The corresponding constraint is 
\begin{equation}
\label{eq: safe SACBF}
\psi_{i,0}^{s}(\boldsymbol{x}_{i,s}) = b_i(\boldsymbol{x}_{i,s}) =\|\boldsymbol{x}_{i,s}-\boldsymbol{x}_{i,s_0}\|_{\bar{p}_i} - r_{i,s_{0}}^{\bar{p}_i} \ge 0, 
\end{equation} 
where $\boldsymbol{x}_{i,s_0}\in\mathbb{R}^{n_i}$, $n_i\le n$, and $\|\cdot\|_{\bar{p}_i}$ denotes the $\bar{p}_i$-norm. Here, the index $i$ denotes the $i$-th safety requirement.

For the \textbf{finite-time reach-and-remain requirement}, we define the states $\boldsymbol{x}_{j,c}\in\mathbb{R}^{n_j}$ that must reach a target region centered at $\boldsymbol{x}_{j,c_d}$ with a constant boundary radius $\varepsilon_{j,c_d}>0$ within the finite time horizon $[T_{j,c_0},T_{j,c_1}]$ by satisfying 
\begin{equation} 
\label{eq: reachability SACBF1}
\psi_{j,0}^{c}(\boldsymbol{x}_{j,c},t) = h_j(\boldsymbol{x}_{j,c},t)=\varepsilon_j(t)^{\bar{p}_j} - \|\boldsymbol{x}_{j,c}-\boldsymbol{x}_{j,c_d}\|_{\bar{p}_j} \ge 0, 
\end{equation}
where $\boldsymbol{x}_{j,c_d}\in\mathbb{R}^{n_j}$, $n_j\le n$, and $\|\cdot\|_{\bar{p}_j}$ denotes the $\bar{p}_j$-norm. The function $\varepsilon_j(t)$ specifies a contracting convergence radius: 
\begin{equation} 
\label{eq: reachability SACBF2}
\varepsilon_j(t) = \varepsilon_{j,c_0} - K_{j,c}(t-T_{j,c_0}),\quad K_{j,c}=\frac{\varepsilon_{j,c_0}-\varepsilon_{j,c_d}}{T_{j,c_1}-T_{j,c_0}}, \end{equation} 
with $\varepsilon_{j,c_0}\ge\varepsilon_{j,c_d}\ge0$. Here, the index $j$ denotes the $j$-th reachability requirement. $K_{j,c}$ denotes the contraction rate of the convergence boundary,
with $\varepsilon_{j,c_0}$
specifying the initial radii. The state $\boldsymbol{x}_{j,c}$ is always contained within the region defined by $\varepsilon_j(t)$.
As $\varepsilon_j(t)$ gradually decreases, the state progressively converges until it reaches the target region within the specified time horizon. For the ``remain'' requirement that the state remains within the radius $\varepsilon_{j,c_d}$ over the interval $[T_{j,c_1},T_{j,c_2}]$, we set $\varepsilon_{j,c_0}=\varepsilon_{j,c_d}$ and $K_{j,c}=0$ for the time horizon $[T_{j,c_1},T_{j,c_2}]$. In this case, the radius remains constant, and the SACBF enforces that the state remains within this fixed region throughout the interval.
\subsection{Computational Complexity}
The computational burden of solving Problem~1 is dominated by the SACBF-QP solved at each sampling instant. 
With $q$ decision variables and $N$ constraints, the QP has a worst-case complexity of 
$\mathcal{O}\big((q+N)^3\big)$. 
The additional SACBF overhead arises from evaluating the Taylor-based upper bound $\bar{M}_k(\Delta t)$ in~\eqref{eq: bound2}, 
which depends on the expression of $\Phi$ in~\eqref{eq: bound1}. 
Since $\bar{M}_k$ is approximated using only $r$ Gauss--Legendre nodes, this step scales linearly as 
$\mathcal{O}(r\,C_\Phi)$, where $C_\Phi$ is the cost of evaluating $\Phi$. 
Therefore, the overall per-step complexity is the sum of the QP complexity and the bound approximation, but it is empirically dominated by QP.

\section{Case Studies}
\label{sec:Case Study and Simulations}
\begin{figure*}[!t]
    \vspace*{0.2cm}
    \centering
    \begin{subfigure}[t]{0.32\linewidth}
        \centering
        \includegraphics[width=1.0\linewidth]{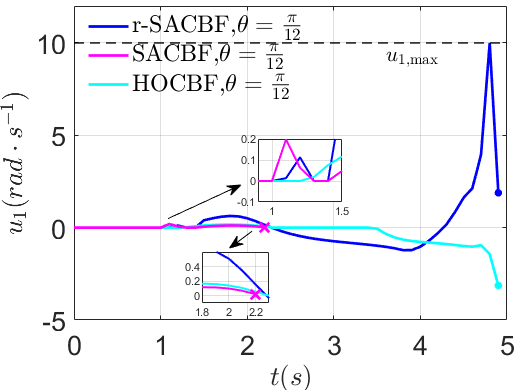}
        \caption{Input $u_{1}$ varies over time.}
        \label{subfig:1}
    \end{subfigure}
    \begin{subfigure}[t]{0.32\linewidth}
        \centering
        \includegraphics[width=1.0\linewidth]{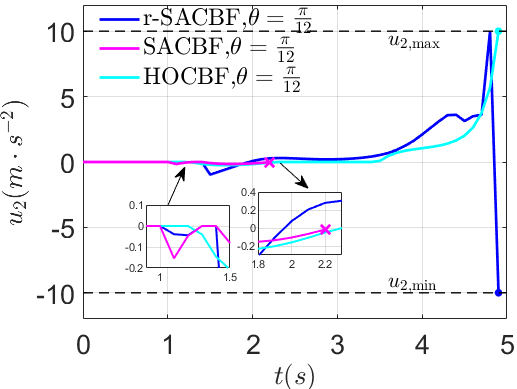}
        \caption{Input $u_{2}$ varies over time.}
        \label{subfig:2}
    \end{subfigure}  
    \begin{subfigure}[t]{0.32\linewidth}
        \centering
        \includegraphics[width=1.0\linewidth]{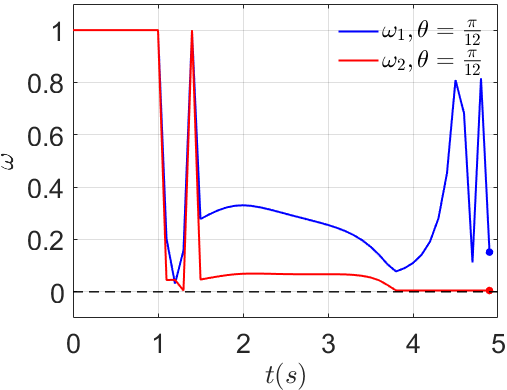}
        \caption{Slack variables vary over time for r-SACBF (blue for safety, red for reachability).}
        \label{subfig:3}
    \end{subfigure}
    \caption{Time evolution of decision variables under r-SACBF, SACBF, and HOCBF methods when initial heading angle is $\frac{\pi}{12}$. The cross symbol ``\(\times\)'' indicates that the QP is infeasible at that time step.}
    \label{fig: decision variable}
\end{figure*}

In this section, we consider a unicycle with dynamics described by:
\begin{small}
\begin{equation}
\label{eq:UM-dynamics1}
\underbrace{\begin{bmatrix}
\dot{x}(t) \\
\dot{y}(t) \\
\dot{\theta}(t)\\
\dot{v}(t)
\end{bmatrix}}_{\dot{\boldsymbol{x}}(t)}  
=\underbrace{\begin{bmatrix}
 v(t)\cos{(\theta(t))}  \\
 v(t)\sin{(\theta(t))} \\
 0 \\
 0
\end{bmatrix}}_{f(\boldsymbol{x}(t))} 
+ \underbrace{\begin{bmatrix}
  0 & 0\\
  0 & 0\\
  1 & 0\\
  0 & 1 
\end{bmatrix}}_{g(\boldsymbol{x}(t))}\underbrace{\begin{bmatrix}
   u_{1}(t)   \\
  u_{2}(t) 
\end{bmatrix}}_{\boldsymbol{u}(t)},
\end{equation}
\end{small}
which is of the form \eqref{eq: affine-system}. In \eqref{eq:UM-dynamics1}, $[x, y]^{\top}$ denote the coordinates of the unicycle, $v$ is its linear speed, $\theta$ denotes the heading angle, and $\boldsymbol{u}$ represent the angular velocity ($u_{1}$) and linear acceleration ($u_{2}$), respectively. The input bounds are defined as $\mathcal{U}=\{\mathbf{u}_{t}\in \mathbb{R}^{2}: -10\cdot \mathcal{I}_{2\times1} \le \mathbf{u}\le 10\cdot \mathcal{I}_{2\times1}\}.$

The QP at each time step is solved using MATLAB’s \texttt{quadprog}, and the system dynamics are integrated with the \texttt{ode45} solver. All simulations are carried out on a computer with an Intel\textsuperscript{\textregistered} Core\texttrademark{} i7-11750F CPU running at 2.50GHz.

\subsection{Safety Requirement}
For the safety requirement, we consider $N_s$ circular obstacles that the robot must avoid, indexed by $i \in\{1,...,N_s\}$.
The safety boundary is constructed from a quadratic distance function defined as
\begin{equation}
\label{eq: task1}
    b_{i}(\boldsymbol{x}_{i,s}) = (x - x_{i,s_{0}})^{2} + (y - y_{i,s_{0}})^{2} - r_{i,s_{0}}^{2},
\end{equation}
where $(x_{i,s_{0}}, y_{i,s_{0}})$ and $r_{i,s_{0}}$ denote the center coordinates and radius of the $i$-th obstacle, respectively. 
Based on \eqref{eq: safe SACBF}, we define SACBF as 
\begin{equation}
\label{eq: SACBF1}
\psi_{i,0}^{s}(\boldsymbol{x}_{i,s}) = b_i(\boldsymbol{x}_{i,s}).
\end{equation}

\subsection{Finite-Time Reach-and-Remain Requirement}
The finite-time reach-and-remain requirement specifies that the system must reach a set of circular target regions, indexed by $j\in\{1,...,N_c\}$, and then remain within them.
The boundary of the $j$-th desired region is constructed from a quadratic distance function defined as
\begin{equation}
\label{eq: task3}
\begin{split}
&h_{j}(\boldsymbol{x}_{j,c})= \varepsilon_{j,c_d}^{2}-(x-x_{j,c_d})^{2}-(y-y_{j,c_d})^{2},\\
\end{split}
\end{equation}
where $(x_{j,c_d}, y_{j,c_d})$ and $\varepsilon_{j,c_d}$ denote the center coordinates and radius of the $j$-th target area, respectively. Based on \eqref{eq: reachability SACBF1}, \eqref{eq: reachability SACBF2}, we define SACBF as 
\begin{equation}
\label{eq: SACBF2}
\begin{split}
\psi_{j,0}^{c}(\boldsymbol{x}_{j,s},t) = \varepsilon_{j,c_0}^{2} - K_{j,c}(t-T_{j,c_0}) -(x-x_{j,c_d})^{2}-\\(y-y_{j,c_d})^{2}, 
K_{j,c}=\frac{\varepsilon_{j,c_0}^{2}-\varepsilon_{j,c_d}^{2}}{T_{j,c_1}-T_{j,c_0}}.
\end{split}
\end{equation}
Here, $\varepsilon_{j,c_0}$ is the initial radius of the circular region containing the robot at time $T_{j,c_0}$, and $K_{j,c}$ specifies the linear shrinking rate that decreases the radius from $\varepsilon_{j,c_0}$ to $\varepsilon_{j,c_d}$ over $[T_{j,c_0},T_{j,c_1}]$. This yields a time-varying boundary that contracts toward the $j$-th target region, ensuring finite-time reachability. If the robot is required to remain within this region over $[T_{j,c_1},T_{j,c_2}]$, we set $\varepsilon_{j,c_0}=\varepsilon_{j,c_d}$ and $K_{j,c}=0$ in~\eqref{eq: SACBF2} for that interval.
\subsection{Complete Cost Function for SACBF-QP}
By formulating the constraints from the SACBFs in \eqref{eq: SACBF}, which are constructed based on \eqref{eq: SACBF1} and \eqref{eq: SACBF2} (with relative degree 2), together with the input bounds \eqref{eq: input bounds}, we can define the cost function for QP as
\begin{equation}
\label{eq:SACBF-QP}
\begin{split}
\min_{\boldsymbol{u}(t),\omega_{k}(t)} \int_{0}^{T}[\boldsymbol{u}(t)^{\top}\boldsymbol{u}(t)+\sum_{k=1}^{N_s+N_c} q_{k}(\omega_{k}(t)-1)^{2}]dt,
\end{split}
\end{equation}
where $\omega_{k}$ is a relaxation variable introduced in the SACBF constraint \eqref{eq: positive guarantee2},  $q_{k}$ is a positive weight scalar, and the slack variable $\omega_{k}$ is designed to converge to $1$. Since existing approaches that address the inter-sampling issue cannot handle high-order constraints, the HOCBF framework \eqref{eq: highest HOCBF} is selected as the benchmark for comparison.
The proposed method without the slack variable is denoted as SACBF, while the one incorporating the relaxation variable is referred to as the relaxed SACBF (r-SACBF).
All methods are implemented with identical hyperparameters for a fair comparison.
\subsection{Case Study \Romannum{1}}
\begin{figure}[ht]
\vspace*{3mm}
    \centering
    \includegraphics[scale=0.18]{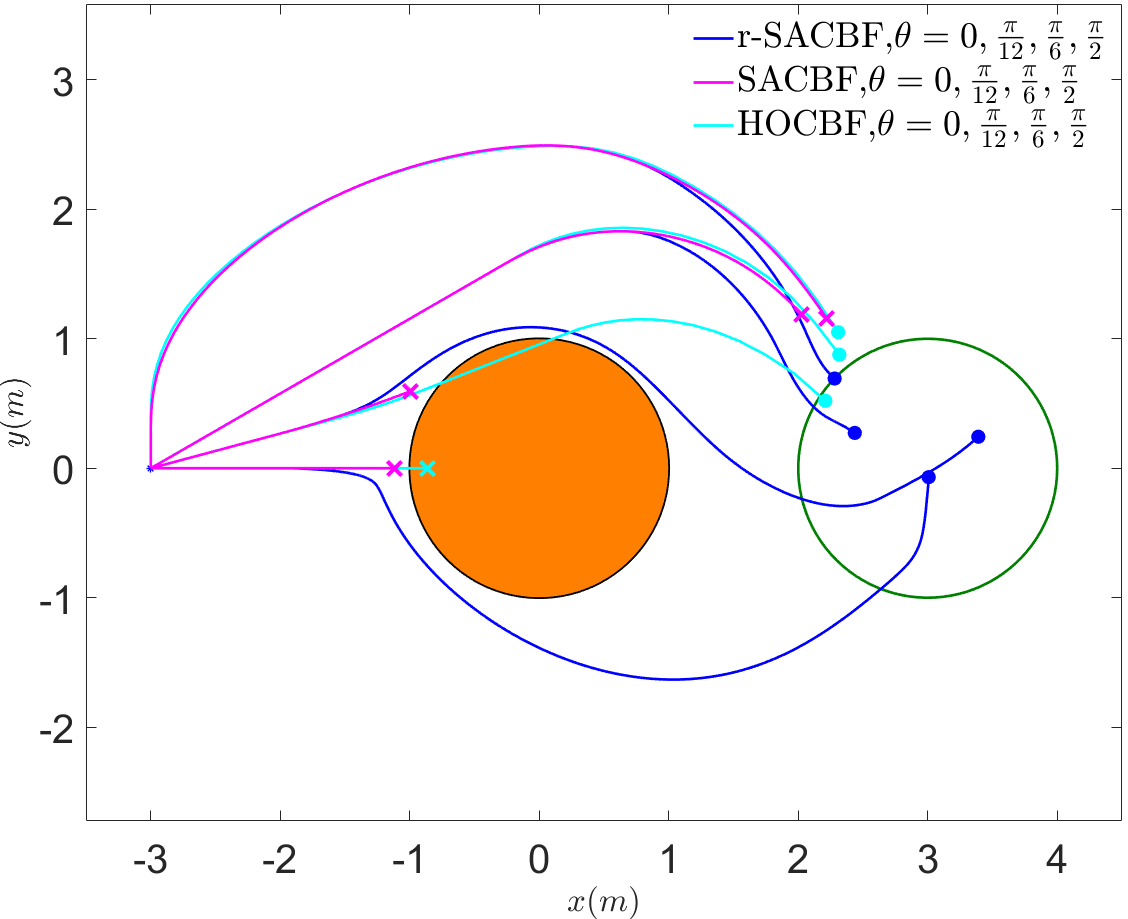}
    \caption{Closed-loop trajectories with controllers derived using r-SACBF (blue), SACBF (magenta) and HOCBF (cyan) with different initial heading angle. The cross symbol ``\(\times\)'' indicates that the QP is infeasible at that location.}
    \label{fig:SACBF-trajectory}
\end{figure} 
Consider a case where the robot must avoid a circular obstacle centered at $(0,0)$ with a radius of $1$, while reaching a circular target region centered at $(3,0)$ with a radius of $1$ within $5$ seconds. The robot starts from the initial position $(-3,0)$ with an initial speed of $1$. The hyperparameters are set as $\lambda_{k}=2$, $\eta_{k}=1$, $q_{k}=200$, $r=5$, $\varepsilon_{1,c_0}=7$, $T_{1,c_1}=5$, $T_{1,c_0}=0$, and $\Delta t=0.1$.

As shown in Fig.~\ref{fig:SACBF-trajectory}, a smaller initial heading angle makes it more difficult for the robot to satisfy the safety requirement. The constraint in \eqref{eq: SACBF1} enforces safety by steering the robot away from obstacles, while the constraint in \eqref{eq: SACBF2} enforces finite-time reachability by driving it toward the target region. Since the robot’s initial position, the obstacle center, and the target center are collinear, a smaller heading angle (i.e., the initial velocity points more directly toward the obstacle) intensifies the conflict between safety and reachability.
This increased conflict may lead to unsafe behaviors or even infeasibility of the optimization problem. Due to the inter-sampling issue, the HOCBF method becomes unsafe when the initial heading angles are $0$ or $\frac{\pi}{12}$, as the robot crosses the obstacle boundary. In contrast, the SACBF and r-SACBF methods maintain safety under the same conditions. The SACBF method incorporates a Taylor-based upper bound to address the inter-sampling issue, which reduces the feasible space of the control input and makes infeasibility more likely compared to HOCBF. Consequently, when the initial heading angles are $0$ or $\frac{\pi}{12}$, the SACBF method becomes infeasible due to overly restrictive safety constraints, whereas for initial heading angles of $\frac{\pi}{6}$ or $\frac{\pi}{2}$, infeasibility arises from excessively tight reachability constraints. Meanwhile, the r-SACBF method achieves safe and successful convergence to the desired region for all tested initial heading angles.

Fig.~\ref{fig: decision variable} illustrates the time evolution of the decision variables when the initial heading angle is $\frac{\pi}{12}$.
From Figs.~\ref{subfig:1} and \ref{subfig:2}, it can be observed that the control inputs of the r-SACBF and SACBF methods respond earlier than those of the HOCBF method to satisfy the safety and reachability constraints.
Moreover, due to the introduction of relaxation variables, the r-SACBF method avoids infeasibility.
As shown in Fig.~\ref{subfig:3}, the relaxation variables decrease toward zero when the corresponding constraints become tight, thereby enhancing feasibility, which is consistent with the description provided in Sec.~\ref{subsec:relaxed SACBF}.

\subsection{Case Study \Romannum{2}}
In this case, we consider a more complex task to evaluate the performance of the r-SACBF method under multiple constraints.
The robot is required to avoid three circular obstacles with radii of $1$, $2$, and $1.5$, centered at $(0,0)$, $(8,0)$, and $(3,3)$, respectively.
The mission specifies that the robot must reach desired region 1 (centered at $(3,0)$ with a radius of $1$) within $5$ seconds, remain inside this region for the following $7$ seconds, then reach region 2 (centered at $(0,7)$ with a radius of $0.5$) within the next $6$ seconds, and finally reach region 3 (centered at $(4,5)$ with a radius of $0.2$) within the subsequent $4$ seconds.
The robot starts from the initial position $(-3,0)$ with an initial speed of $1$ and an initial heading angle of $\frac{\pi}{6}$. The hyperparameters are set as $\lambda_{k}=2$, $\eta_{k}=1$, $q_{k}=1$, $r=5$, $\varepsilon_{1,c_0}=7$, $T_{1,c_2}=12$, $T_{1,c_1}=5$, $T_{1,c_0}=0$, $\varepsilon_{2,c_0}=9$, $T_{2,c_1}=18$, $T_{2,c_0}=12$, $\varepsilon_{3,c_0}=6$, $T_{3,c_1}=22$, $T_{3,c_0}=18$ and $\Delta t=0.1$.
\begin{figure}[ht]
\vspace*{3mm}
    \centering
\includegraphics[scale=0.36]{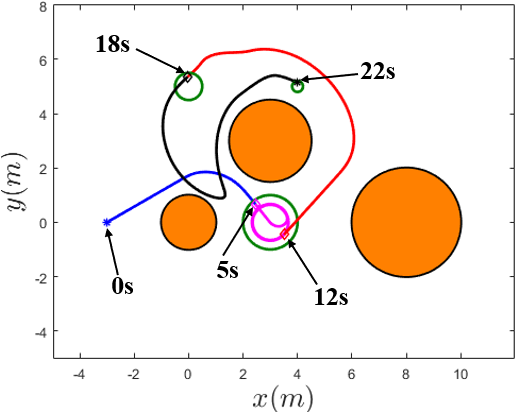}
    \caption{Closed-loop trajectory generated by the r-SACBF method for the multi-constraint task. 
The robot avoids three obstacles and reaches the target regions at 5s, 18s, and 22s as indicated.}
    \label{fig:STLtasks2}
\end{figure} 

Fig.~\ref{fig:STLtasks2} shows the robot trajectory from $0$s to $22$s under the r-SACBF framework.
As illustrated by the color-coded trajectories corresponding to different time intervals, the robot remains feasible throughout the task despite the presence of multiple constraints (three safety constraints, one reach-and-remain constraint, and one input bound).
In particular, as the target region gradually shrinks, the robot is still able to reach it within the required time, demonstrating the effectiveness and flexibility of the r-SACBF in handling complex, multi-constraint tasks. Fig.~\ref{fig:STLSACBFs} illustrates the time evolution of the SACBF functions.
Since the SACBFs have a relative degree of $2$ with respect to system~\eqref{eq:UM-dynamics1}, both the zero-order SACBFs and the first-order SACBFs are shown. It can be observed that, under the enforcement of \eqref{eq: SACBF}, all functions remain positive within the 22s horizon, indicating that both safety and reach-and-remain requirements are satisfied.
\begin{figure}[ht]
\vspace*{3mm}
    \centering
    \includegraphics[scale=0.115]{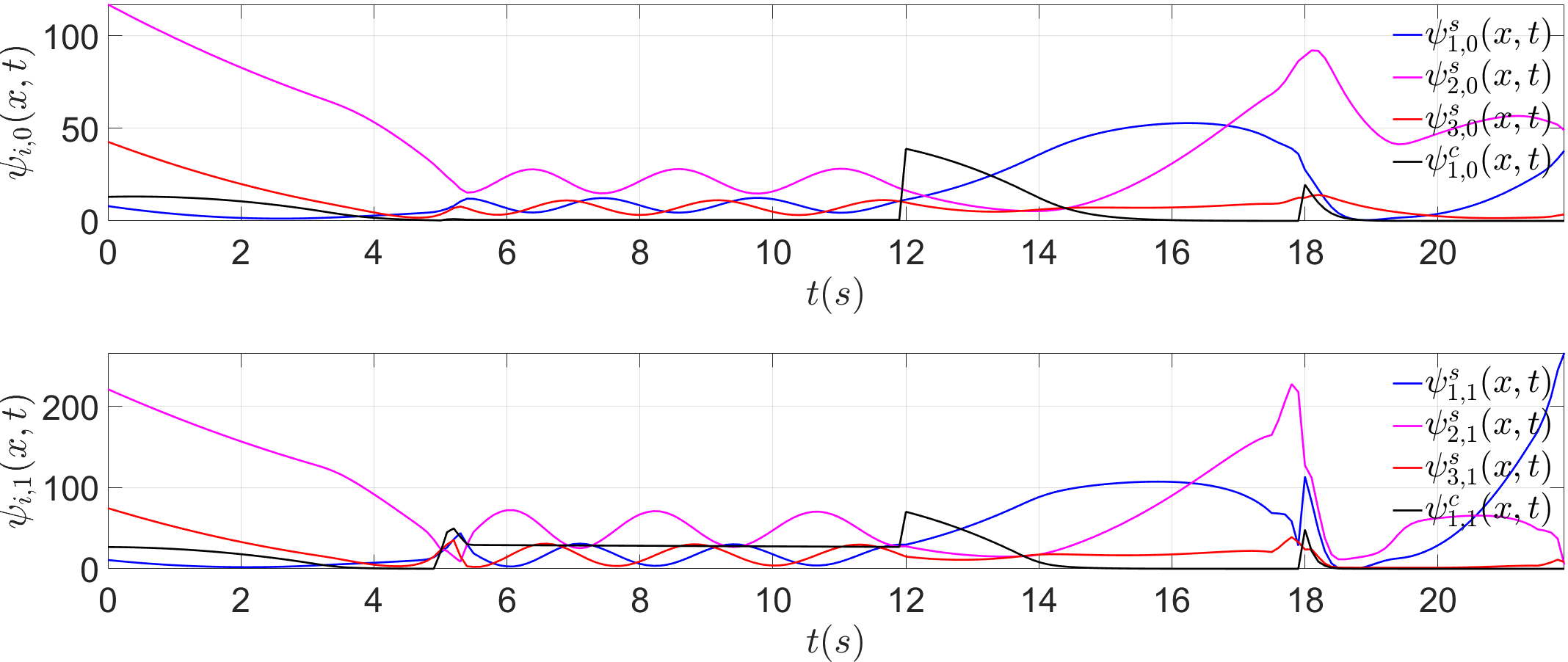}
    \caption{Time evolution of the SACBF functions.
The top plot shows the zero-order SACBFs, and the bottom plot shows the first-order SACBFs.
The four curves correspond to three safety constraints (blue, magenta, red) and one reach-and-remain constraint (black).
All functions stay positive over the 22s horizon.}
    \label{fig:STLSACBFs}
\end{figure} 

\section{Conclusion and Future Work}
This paper introduced Sampling-Aware Control Barrier Functions (SACBFs) to ensure safety and finite-time reach-and-remain for nonlinear systems with high relative-degree constraints under sampled implementations. By incorporating Taylor-based upper bounds, the SACBF framework explicitly captures inter-sampling effects and guarantees continuous-time forward invariance of the desired sets under zero-order-hold control. The introduction of relaxation variables further improves feasibility when multiple constraints are jointly enforced. Simulation studies on a unicycle robot validated the effectiveness and flexibility of the proposed approach, demonstrating safe and feasible performance in complex multi-constraint scenarios where conventional HOCBF-based methods fail. Future work will focus on extending SACBFs to systems with stochastic disturbances and model uncertainties, and integrating learning-based techniques for adaptive estimation of the Taylor-based upper bounds to further reduce conservatism and improve real-time scalability.
\label{sec:conclusion}
\bibliographystyle{IEEEtran}
\balance
\bibliography{references.bib}
\end{document}